\documentclass{article}

\usepackage{epsfig,amsmath,amstext,psfrag,mathrsfs,amssymb,amsfonts,color,amsthm}

\usepackage{latexsym}
\usepackage{enumerate,cite}
\usepackage{graphicx}

\DeclareMathOperator*{\argmax}{arg\,max}

\newcommand{\mmd}{\text{MMD}}

\newtheorem{lemma}{\textbf{Lemma}}
\newtheorem{theorem}{\textbf{Theorem}}

\newtheorem{remark}{\textbf{Remark}}

\newcommand{\nn}{\nonumber}

\newcommand{\mE}{\mathbb{E}}

\newcommand{\cH}{\mathcal{H}}

\newcommand{\cP}{\mathcal{P}}

\DeclareMathAlphabet{\matheuf}{U}{euf}{m}{n}

\begin{document}


\begin{center}
  \baselineskip 1.3ex {\Large \bf Universal Outlying sequence detection For Continuous Observations
}
\\
\vspace{0.15in} Yuheng Bu$^{\star }$ \qquad Shaofeng Zou $^{\dagger}$  \qquad Yingbin Liang $^{\dagger}$ \qquad Venugopal V. Veeravalli $^{\star}$\\
\vspace{0.15in}
$^{\star}$ University of Illinois at Urbana-Champaign\\  $^{\dagger}$ Syracuse University\\
\vspace{0.15in} Email: \small{ \texttt{bu3@illinois.edu, szou02@syr.edu, yliang06@syr.edu, vvv@illinois.edu}}

\end{center}

%
\begin{abstract}
  The following detection problem is studied, in which there are $M$ sequences of samples out of which one outlier sequence needs to be detected. Each typical sequence contains $n$ independent and identically distributed (i.i.d.) \emph{continuous} observations from a known distribution $\pi$, and the outlier sequence contains $n$ i.i.d. observations from an outlier distribution $\mu$, which is distinct from $\pi$, but otherwise unknown. A universal test based on KL divergence is built to approximate the maximum likelihood test, with known $\pi$ and unknown $\mu$. A data-dependent partitions based KL divergence estimator is employed. Such a KL divergence estimator is further shown to converge to its true value exponentially fast when the density ratio satisfies $0<K_1\leq \frac{d\mu}{d\pi}\leq K_2$, where $K_1$ and $K_2$ are positive constants, and this further implies that the test is exponentially consistent. The performance of the test is compared with that of a recently introduced test for this problem based on the machine learning approach of maximum mean discrepancy (MMD). We identify regimes in which the KL divergence based test is better than the MMD based test.
\end{abstract}
%

\section{Introduction}

\label{sec:intro}
In this paper, we study problem, in which there are $M$ sequences of samples out of which one outlier sequence needs to be detected. Each typical sequence consists of $n$ independent and identically (i.i.d.) \emph{continuous} observations drawn from a known distribution $\pi$, whereas the outlier sequence consists of $n$ i.i.d. samples drawn from a distribution $\mu$, which is distinct from $\pi$, but otherwise unknown. The goal is to design a test to detect the outlier sequence.


The study of such a model is very useful in many applications \cite{vvv2014}. For example, in cognitive wireless networks, signals follow different distributions depending on whether the channel is busy or vacant.  The goal in such a network is to identify vacant channels out of busy channels based on their corresponding signals in order to utilize the vacant channels for improving spectral efficiency. Such a problem was studied in \cite{Lai2011} and \cite{Tajer2013} under the assumption that both $\mu$ and $\pi$ are known. Other applications include anomaly detection in large data sets \cite{bolton2002statistical,Chan2009}, event detection and environment monitoring in sensor networks \cite{chamberland2007wireless}, understanding of visual search in humans and animals \cite{vaidhiyan2012active}, and optimal search and target tracking \cite{stone1975theory}.

The outlying sequence detection problem with \emph{discrete} $\mu$ and $\pi$ was studied in \cite{Li2013}. A universal test based on generalized likelihood ratio test was proposed, and was shown to be  exponentially consistent. The error exponent was further shown to be optimal as the number of sequences goes to infinity. The test utilizes empirical distributions to estimate $\mu$ and $\pi$, and is therefore applicable only for the case where $\mu$ and $\pi$ are discrete.

In this paper, we study the case where distributions $\mu$ and $\pi$ are \emph{continuous} and $\mu$ is \emph{unknown}. We construct a Kullback-Leibler (KL) divergence  based test, and further show that this test is \emph{exponentially consistent}.


Our exploration of the problem starts with the case in which both $\mu$ and $\pi$ are known, and the maximum likelihood test is optimal. An interesting observation is that the test statistic of the optimal test converges to $D(\mu \| \pi)$ as the sample size goes to infinity if the sequence is the outlier. This motivates the use of a KL divergence estimator to approximate the test statistic for the case when $\mu$ is unknown. We apply a divergence estimator based on the idea of data-dependent partitions \cite{wang2005divergence}, which was shown to be consistent. Our first contribution here is to show that such an estimator converges exponentially fast to its true value when the density ratio satisfies the boundedness condition: $0< K_1\leq \frac{d\mu}{d\pi}\leq K_2$, where $K_1$ and $K_2$ are positive constants. We further design a KL divergence based test using such an estimator and show that the test is exponentially consistent.




The rest of the paper is organized as follows. In Section \ref{sec:format}, we describe the problem formulation. In Section \ref{sec:test}, we present the KL divergence based test and establish its exponential consistency. In Section \ref{sec:MMD}, we review the maximum mean discrepancy (MMD) based test. In Section \ref{sec:numerical}, we provide a numerical comparison of our KL divergence based test and the MMD based test. All the detailed proofs is shown in the appendix.

\section{Problem Model}

\label{sec:format}
Throughout the paper, random variables are denoted
by capital letters, and their realizations are denoted
by the corresponding lower-case letters. All logarithms are with respect to the natural base.

We study an outlier detection problem, in which there are in total $M$ data
sequences denoted by $Y^{(i)}$ for $1 \le i \le M$. Each data sequence $Y^{(i)}$ consists of $n$ i.i.d. samples
$Y^{(i)}_1, \dots , Y^{(i)}_n$ drawn from either a typical distribution $\pi$ or an outlier distribution $\mu$, where $\pi$ and $\mu$ are \emph{continuous}, i.e.,   defined on $(\mathbb{R},\mathcal{B}_{\mathbb{R}})$, and $\mu \ne \pi$. We use the notation $\mathbf{y}^{(i)} = (y^{(i)}_1, \dots , y^{(i)}_n)$, where $y^{(i)}_k \in \mathbb{R}$ denotes the $k$-th observation of the $i$-th sequence. We assume that there is exactly one outlier among $M$ sequences. If the $i$-th sequence is the outlier, the joint distribution of all the observations is given by
\begin{equation*}
 p_i(y^{Mn})= p_i(\mathbf{y}^{(1)},\dots,\mathbf{y}^{(M)})=\prod_{k=1}^n\Big\{ \mu(y_k^{(i)}) \prod_{j\ne i } \pi(y_k^{(j)}) \Big\}.
\end{equation*}

We are interested in the scenario in which the outlier distributions $\mu$ is unknown a priori, but we know the typical distribution $\pi$ exactly. This is reasonable because in practical scenarios, systems typically start without outliers and it is not difficult to collect sufficient information about $\pi$.

Our goal is to build a distribution-free test to detect the outlier sequence generated by $\mu$. The the test can be captured by a universal rule
$\delta : \pi \times \mathbb{R}^{Mn} \to {1,\dots,M}$, which must not depend on $\mu$.

The maximum error probability, which is a function of  the detector and $(\mu, \pi)$, is defined as
\begin{equation*}
  e(\delta,\pi,\mu) \triangleq \max_{i=1,\dots,M}  \int_{y^{Mn}: \delta(\pi, y^{Mn})\ne i} p_i( y^{Mn}) dy^{Mn},
\end{equation*}
and the corresponding error exponent is defined as
\begin{equation*}
  \alpha(\delta,\pi,\mu) \triangleq \lim_{n\to \infty}-\frac{1}{n}\log e(\delta,\pi,\mu).
\end{equation*}
A test is said to be \emph{universally consistent} if 
\begin{equation*}
  \lim_{n\to \infty}e(\delta,\pi,\mu) =0,
\end{equation*}
for any $\mu\ne\pi$. It is said to be \emph{universally exponentially consistent} if 
\begin{equation*}
  \lim_{n\to \infty} \alpha(\delta,\pi,\mu) >0,
\end{equation*}
for any $\mu\ne\pi$.

\section{KL divergence based test}\label{sec:test}

We first introduce the optimal test when both $\mu$ and $\pi$ are known, which is the maximum likelihood test. We then construct a KL divergence estimator, and prove its exponential consistency. Next, we employ the KL divergence estimator to approximate the test statistics of the optimal test for the outlying sequence detection problem, and construct the KL divergence based test.

\subsection{Optimal test with $\pi$ and $\mu$  known}
If both $\mu$ and $\pi$ are known, the optimal test for the outlying sequence detection problem is  the maximum likelihood  test:
\begin{equation}\label{test1}
    \delta_{\mathrm{ML}}(y^{Mn},\pi,\mu) = \argmax_{1\le i \le M}\  p_i(y^{Mn}).
\end{equation}
By normalizing $p_i(y^{Mn})$ with $\pi (y^{Mn}) $, \eqref{test1} is equivalent to:
\begin{align}
  \delta_{\mathrm{ML}}(y^{Mn},\pi,\mu)   =\argmax_{1\le i \le M}\  \frac{p_i(y^{Mn})}{ \pi (y^{Mn}) } 
      =\argmax_{1\le i \le M} \left\{ \frac{1}{n}\sum_{k=1}^n\log\frac{\mu(y_k^{(i)})}{\pi(y_k^{(i)})}\right\}
    =\argmax_{1\le i \le M}\ L_i \nn.
\end{align}
where
\begin{equation}\label{test11}
  L_i \triangleq   \frac{1}{n}\sum_{k=1}^n\log\frac{\mu(y_k^{(i)})}{\pi(y_k^{(i)})}.
\end{equation}

The following theorem characterizes the error exponent of test \eqref{test1}.
\begin{theorem}\label{thm:optimal} \cite[Theorem 1]{Li2013} Consider the outlying sequence detection problem with both $\mu$ and $\pi$  known. The error exponent for the maximum
likelihood test \eqref{test1} is given by
 $$ \alpha(\delta_{\mathrm{ML}},\pi,\mu) =2B(\pi,\mu),$$
where $B(\pi,\mu)$ is the Bhattacharyya distance between $\mu$ and $\pi$ which is defined as
\begin{equation*}
  B(\pi,\mu) \triangleq -\log \left( \int \mu(y)^{\frac{1}{2}}\pi(y)^{\frac{1}{2}} dy  \right).
\end{equation*}
\end{theorem}

\begin{proof}
See Appendix \ref{app:thm1}.
\end{proof}


Consider $L_i$ defined in \eqref{test11}. If $\mathbf y^{(i)}$ is generated from  $\mu$,
$L_i \to D(\mu||\pi)$  almost surely as $n \to \infty$, by the Law of Large Numbers.
Here, $$D(\mu||\pi) \triangleq \int d\mu\log\frac{d\mu}{d\pi}$$ is the KL divergence between $\mu$ and $\pi$.
Similarly, if $\mathbf y^{(i)}$ is generated from  $\pi$, $ L_j \to -D(\pi||\mu)$  almost surely as $ n \to \infty.$
If $\mathbf y^{(i)}$ is generated from  $\mu$, $L_i$ is an empirical estimate of the KL divergence between $\mu$ and $\pi$. This motivates us to construct a test based on an estimator of KL divergence between $\mu$ and $\pi$, if $\mu$ is unknown.


\subsection{KL divergence estimator}
We introduce a KL divergence estimator of continuous distributions based on data-dependent partitions \cite{wang2005divergence}.

Assume that the distribution $p$ is \emph{unknown} and the distribution $q$ is known, and both $p$ and $q$ are continuous. A sequence of i.i.d. samples $Y \in \mathbb{R}^n$ is generated from $p$. We wish to estimate the KL divergence between  $p$ and  $q$.
We denote the order statistics of $Y$ by $\{Y_{(1)}, Y_{(2)}, \dots, Y_{(n)} \}$ where $Y_{(1)}\le Y_{(2)}\le \dots\le Y_{(n)}$. We further partition the real line  into empirically equiprobable segments as follows:
\begin{equation*}
\begin{split}
    \{I^n_t \}_{t=1,\dots,T_n}=\{(-\infty,Y_{(\ell_n)}],&\ (Y_{(\ell_n)},Y_{(2\ell_n)}], \dots,(Y_{(\ell_n(T_n-1))},\infty)\},
\end{split}
\end{equation*}
where $\ell_n \in \mathbb{N}\le n$ is the number of points in each interval except possibly the last one, and $T_n=\lfloor n/\ell_n \rfloor$ is the number of intervals.
A divergence estimator between the sequence $Y\in \mathbb{R}^n$ and the distribution $\pi$ was proposed in \cite{wang2005divergence}, which is given by
\begin{equation}\label{estimator}
  \hat{D}_n(Y||q) = \sum_{t=1}^{T_n-1}\frac{\ell_n}{n} \log \frac{\ell_n/n }{q(I^n_t)}+ \frac{\epsilon_n}{n} \log \frac{\epsilon_n/n }{q(I^n_{T_n})},
\end{equation}
where $\epsilon_n=(n-\ell_n(T_n-1))$ is the number of points in the last segment.

The consistency of such an estimator was shown in \cite{wang2005divergence}. Here, we characterize the convergence rate by introducing the following boundedness condition on the density ratio between $p$ and $q$, i.e.,
\begin{equation}\label{bddcondition}
0<K_1 \le  \frac{\mathrm{d}p}{\mathrm{d}q} \le K_2,
\end{equation}
where $K_1$ and $K_2$ are positive constants.
In practice, such a boundedness condition is often satisfied, for example, for truncated Gaussian distributions.

The following theorem characterizes a lower bound on the convergence rate of  estimator  \eqref{estimator}.
\begin{theorem} \label{thm:estimator}  If the density ratio between $p$ and $q$ satisfies \eqref{bddcondition}, and  estimator  \eqref{estimator} is applied with $T_n, \ell_n \to \infty$, as $n\to \infty$, then for $\forall \epsilon >0$,
\begin{equation*}
\begin{split}
\lim_{n\to \infty} -\frac{1}{n} \log\left(\mathbb{P} \left\{ \big|\hat{D}_{n}(Y||q)-D(p||q)\big| > \epsilon \right\}\right)
 \ge \frac{1}{32}\frac{K_1^2}{K_2^2} \epsilon^2.
\end{split}
\end{equation*}
\end{theorem}

\begin{proof}
See Appendix \ref{app:thm2}.
\end{proof}

\begin{remark}
The convergence rate of estimator \eqref{estimator} in Theorem \ref{thm:estimator} is equivalent to
\begin{equation*}
 \big| \hat{D}_{n}(Y||q)-D(p||q)\big| = \mathcal{O}_p(n^{-1/2}),\footnote{$X_n=\mathcal O_p(a_n)$: $\forall \epsilon>0$, $\exists M>0$, $P(|\frac{X_n}{a_n}|>M)<\epsilon, \forall n.$}
\end{equation*}
where $\mathcal{O}_p$ denotes ``bounded in probability '' \cite{nguyen2010estimating}. 
\end{remark}

\subsection{Test and performance}
In this subsection, we utilize the estimator based on data-dependent partitions to construct our test.

It is clear that if $Y^{(i)}$ is the outlier, then $\hat{D}_{n}(Y^{(i)}||\pi)$ is a good estimator of $D(\mu||\pi)$, which is a positive constant. On the other hand, if $Y^{(j)}$ is a typical sequence, $\hat{D}_{n}(Y^{(j)}||q)$ should be close to $D(\pi||\pi)=0$. Based on this understanding and the convergence guarantee in Theorem \ref{thm:estimator}, we use $\hat{D}_{n}(Y^{(i)}||\pi)$ in place of $L_i$ in \eqref{test11}, and construct the following test for the outlying sequence detection problem:
\begin{equation}\label{test2}
    \delta_{\mathrm{KL}}(y^{Mn}) = \arg \max_{1\le j \le M} \hat{D}_n(Y^{(j)}||\pi).
\end{equation}

The following theorem provides a lower bound on the error exponent of $\delta_{\mathrm{KL}}$, which further implies that $\delta_{\mathrm{KL}}$ is universally exponentially consistent.
\begin{theorem}\label{thm:KL test}
If the density ratio between $\mu$ and $\pi$ satisfies \eqref{bddcondition},
then  $\delta_{\mathrm{KL}}$ defined in \eqref{test2} is  exponentially consistent, and the error exponent is lower bounded as follows,
\begin{equation}
\begin{split}
\alpha(\delta_{\mathrm{KL}},\pi,\mu)\ge \frac{1}{32}\left(\frac{K_1}{K_1+K_2}\right)^2 D^2(\mu||\pi).
\end{split}
\end{equation}
\end{theorem}
\begin{proof}
See Appendix \ref{app:thm3}.
\end{proof}

\section{MMD-Based Test}
\label{sec:MMD}

In this section, we introduce the MMD based test, which we previously studied in \cite{zou2014unsupervised}. We will compare $\delta_{\mathrm{KL}}$ to the MMD based test.

\subsection{Introduction to MMD}
In this subsection, we briefly introduce the idea of mean embedding of distributions into RKHS \cite{Srip2010} and the metric of MMD. Suppose $\cP$ is a set of probability distributions, and suppose $\cH$ is the RKHS with an associated kernel $k(\cdot,\cdot)$. We define a mapping from $\cP$ to $\cH$ such that each distribution $p\in \cP$ is mapped to an element in $\cH$ as follows
$$\mu_p(\cdot)=\mE_p [k(\cdot,x)]=\int k(\cdot,x)dp(x). $$
Here, $\mu_p(\cdot)$ is referred to as the {\em mean embedding} of the distribution $p$ into the Hilbert space $\cH$. Due to the reproducing property of $\cH$, it is clear that $\mE_p[f]=\langle \mu_p,f \rangle_{\cH}$ for all $f \in \cH$.

In order to distinguish between two distributions $p$ and $q$, Gretton \emph{et al.} \cite{Gretton2012} introduced the following quantity of maximum mean discrepancy (MMD) based on the mean embeddings $\mu_p$ and $\mu_q$ of $p$ and $q$ in RKHS:
\begin{equation}
\mmd[p,q]:=\|\mu_p-\mu_q\|_{\cH}.\nn
\end{equation}

It can be shown that
\[\mmd[p,q]=\sup_{\|f\|_{\cH} \leq 1} \mE_p[f]-\mE_q[f].\]
Due to the reproducing property of kernel, the following is true
\begin{flalign}
\mmd^2[p,q]=\mE[k(X,X')]-2\mE[k(X,Y)]+\mE[k(Y,Y')],\nn
\end{flalign}
where $X$ and $X'$ are independent but have the same distribution $p$, and $Y$ and $Y'$ are independent but have the same distribution $q$. An unbiased estimator of $\mmd^2[p,q]$ based on $q$ and $n$ samples of $X=\{x_1,x_2,\dots,x_n\}$ generated from $p$ is given as follows,
\begin{flalign}
\mmd_u^2[X,q]=\frac{1}{n(n-1)}\sum_{i=1}^{n}\sum_{j\neq i}^{n} k(x_i,x_j)+\mE[k(Y,Y')]-\frac{2}{n}\sum_{i=1}^{n}\mE [k(x_i,Y)],\nn
\end{flalign}
where $Y$ and $Y'$ are independent but have the same distribution $q$.
\subsection{Test and performance}
For each sequence $Y^{(i)}$, we compute $\mmd_u^2[Y^{(i)},\pi]$ for $1\leq i\leq M$. It is clear that if $Y^{(i)}$ is the outlier,  $\mmd_u^2[Y^{(i)},\pi]$ is a good estimator of $\mmd^2[\mu,\pi]$, which is a positive constant. On the other hand, if $Y^{(i)}$ is a typical sequence, $\mmd_u^2[Y^{(i)},\pi]$ should be a good estimator of $\mmd^2[\pi,\pi]$, which is zero. Based on the above understanding, we construct the following test:
\begin{flalign}\label{eq:test_1_withoutref}
    \delta_{\mmd}=\underset{1\leq i\leq M}{\arg\max}\mmd_u^2[Y^{(i)},\pi].
\end{flalign}

The following theorem provides a lower bound on the error exponent of $\delta_\mmd$, and further demonstrates that the test $\delta_\mmd$ is universally exponentially consistent.
\begin{theorem}\label{thm:s1withoutref}
Consider the universal outlying sequence detection problem. Suppose $\delta_\mmd$ defined in \eqref{eq:test_1_withoutref} applies a bounded kernel with $0\leq k(x,y)\leq K$ for any $(x,y)$.
Then, the error exponent is lower bounded as follows,
\begin{flalign}\label{pe1}
  \alpha(\delta_\mmd,\mu,\pi)\geq \frac{\mathrm{MMD}^4[\mu,\pi]}{9K^2}.
\end{flalign}
\end{theorem}

\begin{proof}
See Appendix \ref{app:thm4}.
\end{proof}
\section{Numerical results and Discussion}
\label{sec:numerical}
In this section, we compare the performance of $\delta_{\mathrm{KL}}$ and $\delta_{\mmd}$.

We set the number of sequences $M=5$. We choose the typical distribution $\pi=\mathcal{N}(0,1) $, and choose the outlier distribution $\mu=\mathcal{N}(0,0.2),\mathcal{N}(0,1.2), \mathcal{N}(0,1.8), \mathcal{N}(0,2.0)$, respectively. In Fig.~\ref{fig:a}, Fig.~\ref{fig:b}, Fig.~\ref{fig:c} and Fig.~\ref{fig:d}, we plot the logarithm of the probability of error $\log P_e$ as a function of the sample size $n$. 

It can be seen that for both tests as the number of samples increases, the probability of error converges to zero as the sample size increases. Furthermore, $\log P_e$ decreases with $n$ linearly, which demonstrates the exponential consistency of both $\delta_{\mathrm{KL}}$ and $\delta_{\mmd}$.

By comparing the four figures, it can be seen that as the variance of $\mu$ deviates from the variance of $\pi$, $\delta_{\mathrm{KL}}$ outperforms $\delta_{\mmd}$. The numerical results and theoretical lower bounds on error exponents  give us some intuitions to identify regimes in which one test outperforms the other. As shown above, when the distribution $\mu$ and $\pi$ become more different from each other, $\delta_{\mathrm{KL}}$ will outperform $\delta_{\mmd}$. The reason is that for any pair of distributions, MMD is bounded between $[0, 2K]$, while the KL divergence is not bounded. As the distributions become more different from each other, the KL divergence will increase, and the KL divergence based test will have a larger error exponent than MMD based test.


%

\begin{figure}[!htp]
\center
  \includegraphics[width=8.8cm]{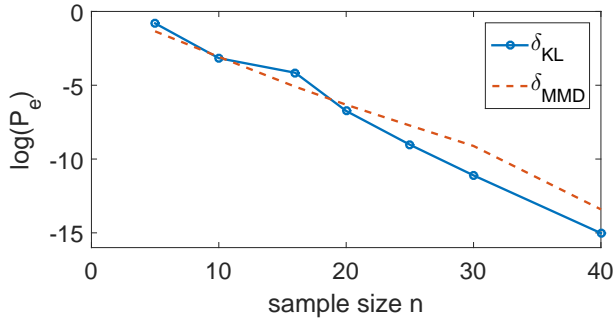}

  \caption{ Comparison of the performance between KL divergence and MMD based test with $\pi=\mathcal{N}(0,1)$ and $\mu=\mathcal{N}(0,0.2)$}\label{fig:a}
\end{figure}
 
\begin{figure}[!htp]
\center
  \includegraphics[width=8.8cm]{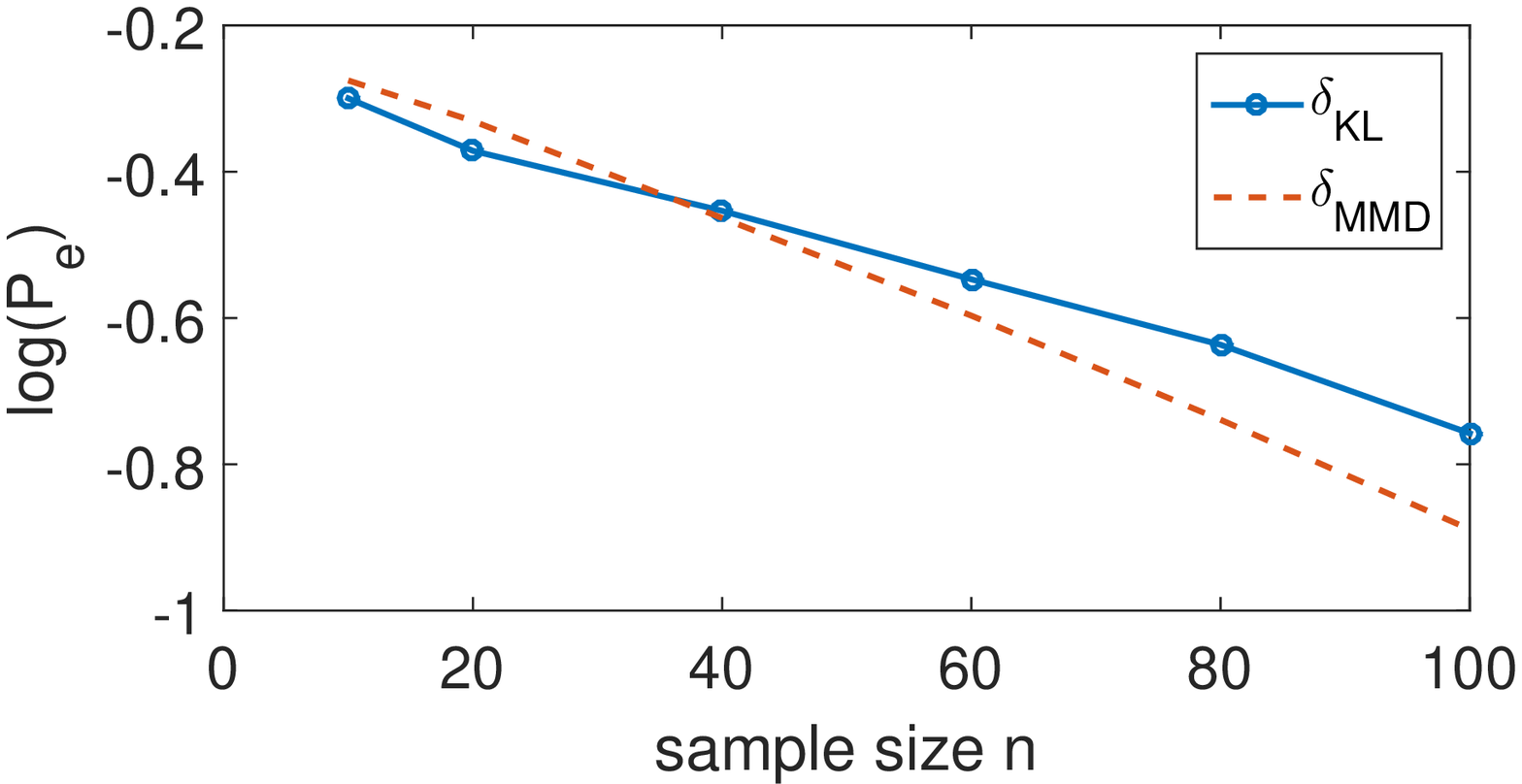}
 
  \caption{ Comparison of the performance between KL divergence and MMD based test with $\pi=\mathcal{N}(0,1)$ and $\mu=\mathcal{N}(0,1.2)$}\label{fig:b}
\end{figure}
\begin{figure}[!htp]
\center
  \includegraphics[width=8.8cm]{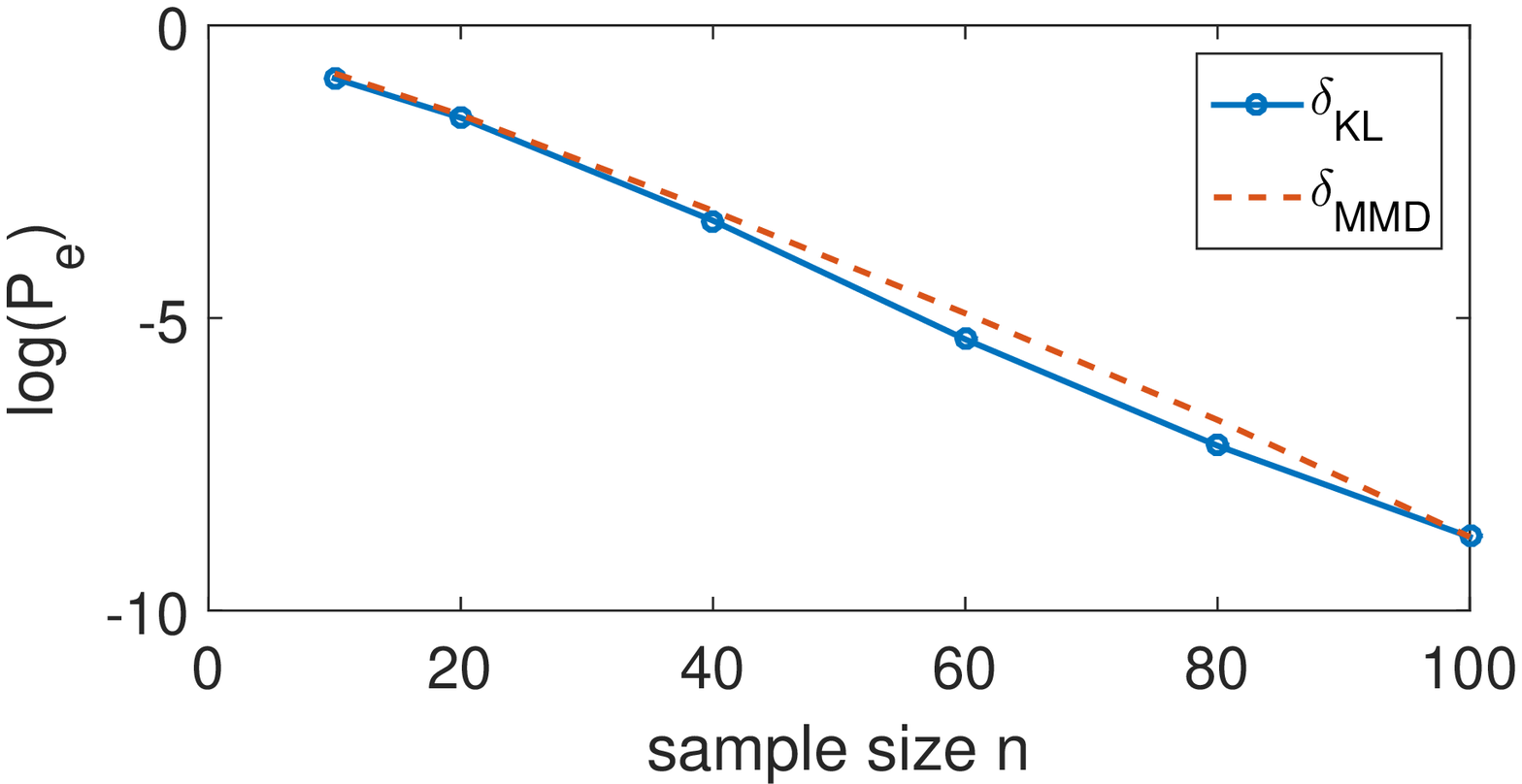}
  \vspace{-.35cm}
  \caption{ Comparison of the performance between KL divergence and MMD based test with $\pi=\mathcal{N}(0,1)$ and $\mu=\mathcal{N}(0,1.8)$}\label{fig:c}
\end{figure}
 \vspace{-.4cm}
\begin{figure}[!htp]
\center
  \includegraphics[width=8.8cm]{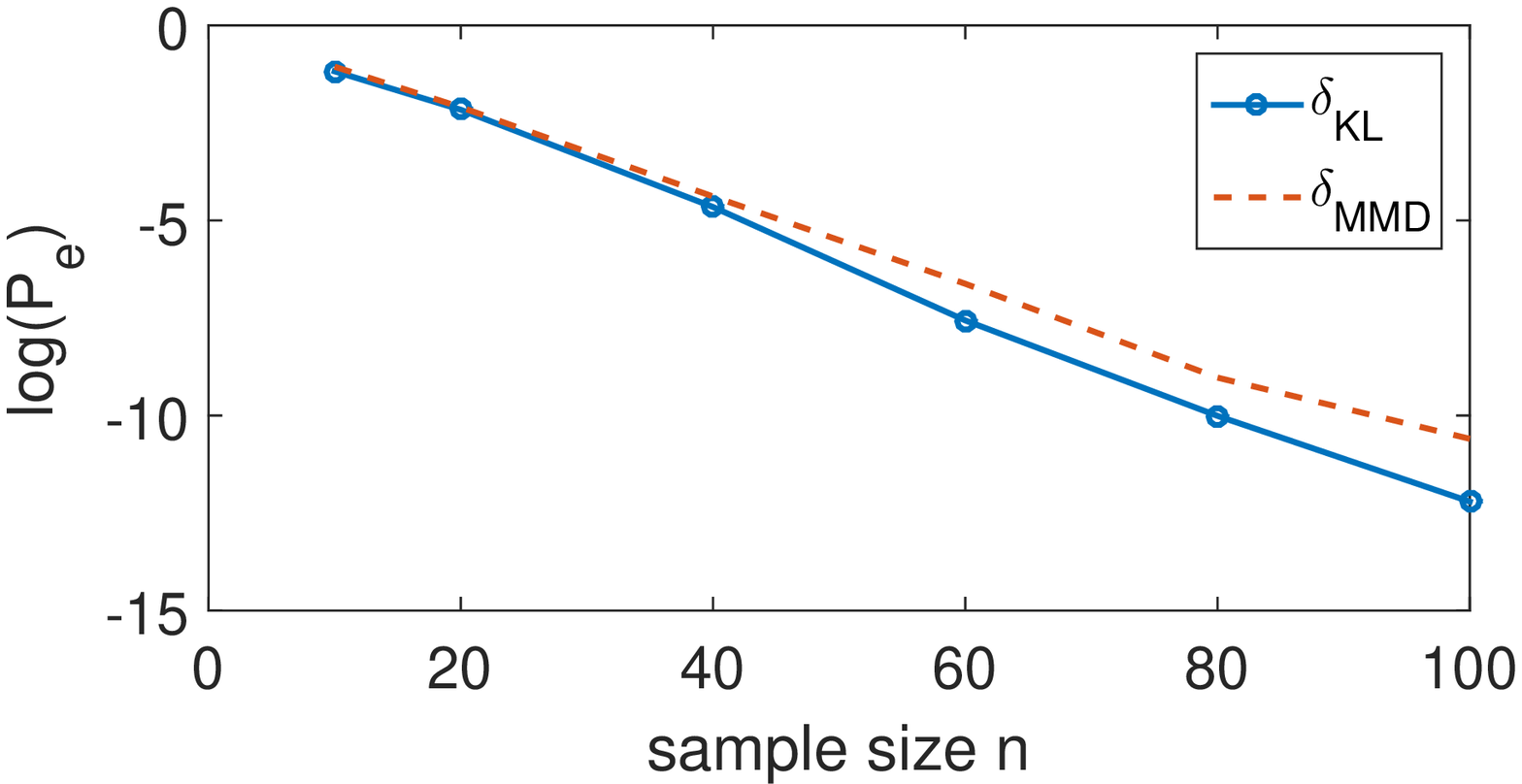}
  \vspace{-.35cm}
  \caption{Comparison of the performance between KL divergence and MMD based test with $\pi=\mathcal{N}(0,1)$ and $\mu=\mathcal{N}(0,2)$}\label{fig:d}
\end{figure}

\newpage
\appendix
\noindent{\Large {\textbf{Appendix}}}
\section{Proof of Theorem 1}
\label{app:thm1}
%
Recall the maximum likelihood test is defined as
\begin{equation*}
  \delta_{\mathrm{ML}}(y^{Mn}) = \argmax_{1\le i \le M} \log \frac{p_i(y^{Mn})}{  \prod_{j=1}^M \prod_{k=1}^n\pi(y_k^{(j)}) }=\argmax_{1\le i \le M} \left\{ \frac{1}{n}\sum_{k=1}^n\log\frac{\mu(y_k^{(i)})}{\pi(y_k^{(i)})}\right\}= \argmax_{1\le i \le M} L_i.
\end{equation*}

Now we will characterize the exponent for the maximum likelihood test. By the symmetry of the problem, it is clear that $\mathbb{P}_i\{\delta \ne i \}$ is the same for every $i = 1,\dots,M$, hence
\begin{equation*}
  \max_{i=1,\dots,M} \mathbb{P}_i\{\delta_{\mathrm{ML}} \ne i \} = \mathbb{P}_1\{\delta_{\mathrm{ML}} \ne 1 \}.
\end{equation*}
It now follows
\begin{equation*}
\mathbb{P}_1\left \{L_1 \le  L_2\right \} \le \mathbb{P}_1\{\delta \ne 1 \} = \mathbb{P}_1\left \{ L_1 \le \max _{2\le j\le M} L_j \right \} \le (M-1) \mathbb{P}_1\left \{L_1 \le  L_2\right \}.
\end{equation*}
Since $\log(M)/n \to 0$, the left hand side and right hand side will share a same error probability exponent, so we just need to compute the exponent for $\mathbb{P}_1\left \{L_1 \le  L_2\right \} $.

Let us use the notation,
\begin{equation*}
  Z_k = \log\left( \frac{\mu(y_k^{(1)})}{\pi(y_k^{(1)})} \frac{\pi(y_k^{(2)})}{\mu(y_k^{(2)})} \right).
\end{equation*}
Then, we can rewrite the probability,
\begin{align}
\mathbb{P}_1\left \{L_1 \le  L_2\right \} &= \mathbb{P}_1\left \{ \sum_{k=1}^n\log\frac{\mu(y_k^{(1)})}{\pi(y_k^{(1)})}-\sum_{k=1}^n\log\frac{\mu(y_k^{(2)})}{\pi(y_k^{(2)})} \le 0  \right \}\nn\\
&= \mathbb{P}_1\left \{ \sum_{k=1}^n Z_k \le 0  \right \}. \nn
\end{align}

Thus we can apply the Cramer's theorem directly.
\begin{align*}
 \lim_{n\to \infty} -\frac{1}{n} \mathbb{P}_1\left \{ \sum_{k=1}^n Z_k \le na \right \} =\Lambda_Z(a),
\end{align*}
for $a<\mathbb{E}(Z)=D(\pi||\mu)+D(\mu||\pi)$, and $\Lambda_Z(a)$ is the large-deviation rate function.

In our case, $0< \mathbb{E}(Z)$ for $\mu \ne \pi$. So
\begin{align*}
 \lim_{n\to \infty} -\frac{1}{n} \mathbb{P}_1\left \{ \sum_{k=1}^n Z_k \le 0 \right \} =\Lambda_Z(0)= \sup_\lambda \big[-\kappa_Z(\lambda) \big].
\end{align*}
We just need to compute the log-MGF of random variable $Z$,
\begin{equation*}
  \kappa_Z(\lambda) = \log\mathbb{E}(e^{\lambda Z} ) =\log \mathbb{E} \left[ \frac{\mu^\lambda(Y^{(1)})}{\pi^\lambda(Y^{(1)})} \frac{\pi^\lambda(Y^{(2)})}{\mu^\lambda(Y^{(2)})} \right].
\end{equation*}
Given that $Y^{(1)}$ is generated from $\mu$, $Y^{(2)}$ is generated from $\pi$, we have
\begin{align}
  \kappa_Z(\lambda) &= \log \left(\int \frac{\mu^{\lambda+1}(y^{(1)})}{\pi^\lambda(y^{(1)})} \frac{\pi^{\lambda+1}(y^{(2)})}{\mu^\lambda(y^{(2)})} dy^{(1)} dy^{(2)} \right)\nn\\
  &=\log \left(\int \frac{\mu^{\lambda+1}(y^{(1)})}{\pi^\lambda(y^{(1)})}   dy^{(1)} \right) + \log \left(\int \frac{\pi^{\lambda+1}(y^{(2)})}{\mu^\lambda(y^{(2)})}  dy^{(2)} \right)\nn\\
  &= -C_\lambda(\pi,\mu)-C_\lambda(\mu,\pi), \nn
\end{align}
where $$C_\lambda(p,q) \triangleq - \log \left( \int p^{\lambda}(y) q^{1-\lambda}(y) dy \right).$$
In this case, it is easy to show that the error exponent
\begin{equation}\label{ratefunction}
 \sup_\lambda \big[-\kappa_Z(\lambda) \big] = \max_\lambda \big[ C_\lambda(\pi,\mu)+C_\lambda(\mu,\pi)\big].
\end{equation}
Since $C_\lambda(p,q)$ is concave with $\lambda$, and $C_\lambda(\pi,\mu) =C_{1-\lambda}(\mu,\pi)$, \eqref{ratefunction} is maximized when $\lambda^* =\frac{1}{2}$, so
\begin{equation*}
 \lim_{n\to \infty} -\frac{1}{n} \max_{i=1,\dots,M} \mathbb{P}_i\{\delta_{\mathrm{ML}} \ne i \}=\max_\lambda \big[ C_\lambda(\pi,\mu)+C_\lambda(\mu,\pi)\big] =2B(\pi,\mu),
\end{equation*}
where $B(\pi,\mu)$ is the Bhattacharyya distance between $\mu$ and $\pi$ which is defined as
\begin{equation*}
  B(\pi,\mu) \triangleq -\log \left( \int \mu(y)^{\frac{1}{2}}\pi(y)^{\frac{1}{2}} dy  \right).
\end{equation*}

\section{Proof of Theorem 2}
\label{app:thm2}
To show the exponential consistency of our estimator, we invoke a result by Lugosi and Nobel \cite{lugosi1996consistency}, that specifies sufficient conditions on the partition of the space under which the empirical measure converges to the true measure.

Let $\mathcal{A}$ be a family of partitions of $\mathbb{R}$. The maximal cell count of $\mathcal{A}$ is given by
\begin{equation*}
  c(\mathcal{A}) \triangleq \sup_{\pi \in \mathcal{A}} |\pi|,
\end{equation*}
where $|\pi|$ denotes the number of cells in partition $\pi$.

The complexity of $\mathcal{A}$ is measured by the growth function as described below. Fix $n$ points in $\mathbb{R}$,
\begin{equation*}
  x_1^n =\{x_1,\dots,x_n  \}.
\end{equation*}
Let $\Delta(\mathcal{A},x_1^n)$ be the number of distinct partitions
\begin{equation*}
  \{I_1\cap x_1^n,\dots, I_r\cap x_1^n  \}
\end{equation*}
of the finite set $x_1^n$ that can be induced by partitions $\pi = \{I_1,\dots,I_r \} \in \mathcal{A}$. Define the growth function of $\mathcal{A}$ as
\begin{equation*}
  \Delta_n^*(\mathcal{A})\triangleq \max_{x_1^n \in \mathbb{R}^n} \Delta(\mathcal{A},x_1^n),
\end{equation*}
which is the largest number of distinct partitions of any $n$-point
subset of $\mathbb{R}$ that can be induced by the partitions in $\mathcal{A}$.

\begin{lemma}\label{Lugosilemma}(Lugosi and Nobel )
Let $Y_1,Y_2,\dots$ be i.i.d. random variables in $\mathbb{R}$ with $Y_i \sim \mu$ and let $\mu_n$ denote the empirical probability measure based on $n$ samples. $\mathcal{A}$ be any collection of partitions of $\mathbb{R}$. For each $n\ge1$ and every $\epsilon >0$, then
\begin{equation}\label{Lugosilemmaeq}
  \mathbb{P} \left\{ \sup_{\pi \in \mathcal{A}} \sum_{I \in \pi} |\mu_n(I)-\mu(I)|>\epsilon \right\}\le 4 \Delta_{2n}^*(\mathcal{A})2^{c(\mathcal{A})}\exp(-n \epsilon^2/32).
\end{equation}
\end{lemma}

To prove theorem 2, we consider the case when typical distribution $q$ is known, and a given sequence $Y \in \mathbb{R}^n$ is independently generated from an unknown distribution $p$. We further assume that $p$ and $q$ are both absolutely continuous probability measures defined on $(\mathbb{R},\mathcal{B}_{\mathbb{R}})$, and satisfy
$$ 0<K_1 \le  \frac{dp}{dq} \le K_2.$$

Denote the empirical probability measure based on the sequence $Y$ by $p_n$ (Since $Y$ is generated from $p$) and defined the empirical equiprobable partitions as follow.
If the order statistics of $Y$ can be expressed as $\{Y_{(1)}, Y_{(2)}, \dots, Y_{(n)} \}$ where $Y_{(1)}\le Y_{(2)}\le \dots\le Y_{(n)}$. The real line is partitioned into empirically equivalent segments according to
\begin{equation*}
\begin{split}
    \{I^n_t \}_{t=1,\dots,T_n}=\{(-\infty,Y_{(\ell_n)}],&\ (Y_{(\ell_n)},Y_{(2\ell_n)}],  \dots,(Y_{(\ell_n(T_n-1))},\infty)\},
\end{split}
\end{equation*}
where $\ell_n \in \mathbb{N}\le n$ is the number of points in each interval except possibly the last one, and $T_n=\lfloor n/\ell_n \rfloor$ is the number of intervals. Assume that as $n \to \infty$, both $T_n, \ell_n \to \infty$. So our estimator can be written as
\begin{equation*}
  \hat{D}_n(Y||q)=\sum_{t=1}^{T_n} p_n(I_t^n) \log\frac{p_n(I_t^n)}{q(I_t^n)}.
\end{equation*}

If we denote the true equiprobable partitions based on true distribution $p$ by $I_t$, then
\begin{equation*}
  p(I_t)=\frac{1}{T_n}= p_n(I_t^n).
\end{equation*}

The estimation error can be decomposed as
\begin{equation*}
\begin{split}
    |\hat{D}_{n}(Y||q)-D(p||q)| \le & \bigg|  \sum_{t=1}^{T_n} p_n(I_t^n) \log \frac{p_n(I_t^n)}{q(I_t^n)}-\sum_{t=1}^{T_n} p(I_t) \log \frac{p(I_t)}{q(I_t)} \bigg|  \\
    & + \bigg|  \sum_{t=1}^{T_n} p(I_t) \log \frac{p(I_t)}{q(I_t)}-\int_{\mathbb{R}} dp \log \frac{dp}{dq}\bigg|
    \triangleq e_1+e_2.
\end{split}
\end{equation*}
Intuitively, $e_2$ is the approximation error caused by numerical integration, which diminishes as $T_n$ increases; $e_1$ is the estimation
error caused by the difference of the empirical equivalent partitions from the true equiprobable partitions and the difference
of the empirical probability measure on an interval from its true probability measure.

In addition, $e_2$ is only depends on $T_n$ and distribution $p$ and $q$, namely, $e_2$ is a deterministic term, while $e_1$ also depends on data $Y$, which is random. Next, we will focus on bounding the $e_1$ term.

Since $ p(I_t)=\frac{1}{T_n}= p_n(I_t^n)$, the approximation error $e_1$ can be written as
\begin{equation*}
\begin{split}
    e_1&=\bigg|  \sum_{t=1}^{T_n} p_n(I_t^n) \log \frac{p_n(I_t^n)}{q(I_t^n)}-\sum_{t=1}^{T_n} p(I_t) \log \frac{p(I_t)}{q(I_t)} \bigg|\\
    &=\bigg|  \sum_{t=1}^{T_n} \frac{1}{T_n} \big( \log q(I_t) -  \log q(I_t^n)\big) \bigg|\\
    & \le \sum_{t=1}^{T_n} \frac{1}{T_n} \bigg|    \big( \log q(I_t) -  \log q(I_t^n)\big) \bigg|\\
    &\le \sum_{t=1}^{T_n} \frac{1}{T_n} f'(\xi_i)\big|     q(I_t) -  q(I_t^n)\big|,
\end{split}
\end{equation*}
where $f(x)=\log x$, and $f'(x)=1/x$, $\xi$ is a real number between $q(I_t)$ and $q(I_t^n)$. We utilize the mean value theorem to get the last inequality.

Since $\xi \ge \min\{q(I_t),q(I^n_t) \}$, we get
\begin{align}\label{controle1}
    e_1  & \le \frac{1}{T_n}  \sum_{t=1}^{T_n}  \max\{\frac{1}{q(I_t)},\frac{1}{q(I_t^n)} \} \big|q(I_t) -  q(I_t^n)\big|\nn \\
        & \le  \frac{\max_{1\le t \le T_n}\{\frac{1}{q(I_t)},\frac{1}{q(I_t^n)} \}}{T_n} \sum_{t=1}^{T_n} \big|q(I_t) -  q(I_t^n)\big|\nn \\
        & =\frac{1}{\alpha}  \sum_{t=1}^{T_n} \big|q(I_t) - q(I_t^n)\big|,
\end{align}
where $$\alpha = \frac{T_n}{\max_{1\le t\le T_n} \{\frac{1}{q(I_t)},\frac{1}{q(I_t^n)} \} }.$$

To get an exponential bound for $e_1$, we will apply lemma \ref{Lugosilemma} to our problem. For our case, $I_t^n$ are the equivalent segments based on the empirical measure $p_n$. Suppose $\mathcal{A}_n$ is the collection of all the partitions of $\mathbb{R}$ into empirically equiprobable intervals based on $n$ sample points. Then, from \eqref{Lugosilemmaeq}
\begin{align}\label{expbound}
\mathbb{P} \left\{  \sum_{t=1}^{T_n} |p_n(I_t^n)-p(I_t^n)|>\epsilon \right\} & \le
\mathbb{P} \left\{ \sup_{\pi \in \mathcal{A}_n} \sum_{I \in \pi} |p_n(I)-p(I)|>\epsilon \right\}\nn \\
&\le 4 \Delta_{2n}^*(\mathcal{A}_n)2^{c(\mathcal{A}_n)}\exp(-n \epsilon^2/32).
\end{align}

If we want to get a meaningful exponential bound, we still need to verify 2 conditions in our case: as $n \to \infty$,
\begin{equation*}
  \mbox{a)  } n^{-1}c(\mathcal{A}_n) \to 0  \mbox{,   \quad b)  } n^{-1}\log \Delta_{2n}^*(\mathcal{A}_n) \to 0.
\end{equation*}

Here,
\begin{equation*}
  c(\mathcal{A}_n) =\sup_{\pi \in \mathcal{A}_n} |\pi|=T_n.
\end{equation*}
Since $\ell_n=n/T_n \to \infty$ as $n \to \infty$, we have that
\begin{equation*}
  \frac{c(\mathcal{A}_n)}{n}=\frac{1}{\ell_n} \to 0.
\end{equation*}

Next consider the growth function $\Delta_{2n}^*(\mathcal{A}_n)$ which is defined
as the largest number of distinct partitions of any $2n$-point subset
of $\mathbb{R}$ that can be induced by the partitions in $\mathcal{A}_n$. Namely
\begin{equation*}
  \Delta_{2n}^*(\mathcal{A}_n)=\max_{x_1^{2n} \in \mathbb{R}^{2n}} \Delta(\mathcal{A}_n,x_1^{2n}).
\end{equation*}
In our algorithm, the partitioning number $\Delta_{2n}^*(\mathcal{A}_n)$ is the
number of ways that $2n$ fixed points can be partitioned by $T_n$ intervals. Then
\begin{equation*}
  \Delta_{2n}^*(\mathcal{A}_n)= {{2n+T_n}\choose {T_n}}.
\end{equation*}
Let $h$ be the binary entropy function, defined as
\begin{equation*}
  h(x)=-x \log(x)-(1-x)\log(1-x), \mbox{for }\ x \in (0,1).
\end{equation*}
By the inequality $\log {s\choose t} \le s h(t/s)$ , we obtain
\begin{equation*}
  \log \Delta_{2n}^*(\mathcal{A}_n) \le (2n+T_n) h\big(\frac{T_n}{2n+T_n}\big)
  \le 3n h\big(\frac{1}{2 \ell_n}\big).
\end{equation*}
As $\ell_n \to \infty$, the last inequality implies that
\begin{equation*}
  \frac{1}{n}\log \Delta_{2n}^*(\mathcal{A}_n) \to 0.
\end{equation*}

Now, we can conclude that the inequality \eqref{expbound} is actually an exponential bound, the coefficients $\Delta_{2n}^*(\mathcal{A}_n)$ and $2^{c(\mathcal{A}_n)}$ will not influence the exponent.

Since $|p_n(I_t^n)-p(I_t^n)|=|\frac{1}{T_n}-p(I_t^n)|=|p(I_t)-p(I_t^n)|$
and $K_1 \le \frac{dp}{dq} \le K_2$, the following holds
\begin{align}\label{convergeQ}
\mathbb{P} \left\{  \sum_{t=1}^{T_n} |q(I_t^n)-q(I_t)|>\epsilon \right\} & \le
\mathbb{P} \left\{ \sum_{t=1}^{T_n} |p(I_t^n)-p(I_t)|> K_1 \epsilon \right\}\nn\\
&=\mathbb{P} \left\{  \sum_{t=1}^{T_n} |p_n(I_t^n)-p(I_t^n)|>K_1 \epsilon \right\}\nn\\
&\le 4 \Delta^*_{2n}(\mathcal{A}_n)2^{c(\mathcal{A}_n)}\exp(-n K_1^2 \epsilon^2/32).
\end{align}

Combine with \eqref{controle1}, we can control the estimation error $e_1+e_2$ with the following bound
\begin{equation*}
\begin{split}
\mathbb{P} \left\{  e_1+e_2 >\epsilon \right\} & \le
\mathbb{P} \left\{ \frac{1}{\alpha}  \sum_{t=1}^{T_n} \big|q(I_t) - q(I_t^n)\big|>\epsilon-e_2 \right\}\\
&\le 4 \Delta^*_{2n}(\mathcal{A}_n)2^{c(\mathcal{A}_n)}\exp(-n \alpha^2 K_1^2 (\epsilon-e_2)^2/32).\\
\end{split}
\end{equation*}
Recall that $$\alpha = \frac{T_n}{\max_{1\le t\le T_n} \{\frac{1}{q(I_t)},\frac{1}{q(I_t^n)} \} }.$$
Since we show that $q(I_t^n)$ converges to $q(I_t)$ exponentially fast in \eqref{convergeQ}, we have
\begin{equation*}
\begin{split}
\lim_{n\to \infty} \alpha & =\lim_{n\to \infty}  \frac{T_n}{\max_{1\le t\le T_n} \{\frac{1}{q(I_t)},\frac{1}{q(I_t^n)} \} }\\
&=\lim_{n\to \infty}  \frac{1}{p(I_t)\max_{1\le t\le T_n} \{\frac{1}{q(I_t)} \} }\\
&=\lim_{n\to \infty}  \frac{\min_{1\le t\le T_n} \{{q(I_t)} \}}{p(I_t) }\ge \frac{1}{K_2}.\\
\end{split}
\end{equation*}
Finally, we can compute the error exponent,
\begin{equation*}
\begin{split}
\lim_{n\to \infty} -\frac{1}{n} \log\left(\mathbb{P} \left\{ |\hat{D}_{n}(Y||q)-D(p||q)| >\epsilon \right\}\right) & \ge
\lim_{n\to \infty} -\frac{1}{n} \log(\mathbb{P} \left\{ e_1+e_2 >\epsilon \right\})\\
&\ge\lim_{n\to \infty} -\frac{1}{n} \log\left\{ 4 \Delta^*_{2n}(\mathcal{A})2^{c(\mathcal{A})}\exp(-n \alpha^2 K_1^2 (\epsilon-e_2)^2/32)\right\})\\
&=  \lim_{n\to \infty} \left(\alpha^2 K_1^2 (\epsilon-e_2)^2/32-\frac{1}{n}\log \Delta_{2n}^*(\mathcal{A}_n) -\frac{c(\mathcal{A}_n)}{n}\right) \\
&=  \lim_{n\to \infty} \frac{\alpha^2 K_1^2 (\epsilon-e_2)^2}{32} \\
&\ge \lim_{n\to \infty} \frac{1}{32}\left(\frac{K_1}{K_2}\right)^2 (\epsilon-e_2)^2. \\
\end{split}
\end{equation*}

Since $e_2$ is the approximation error caused by numerical integration, $\lim_{n \to \infty}e_2= 0$. We prove that
\begin{equation*}
\begin{split}
\lim_{n\to \infty} -\frac{1}{n} \log\left(\mathbb{P} \left\{ |\hat{D}_{n}(Y||q)-D(p||q)| >\epsilon \right\}\right) & \ge
 \frac{1}{32}\left(\frac{K_1}{K_2}\right)^2 \epsilon^2.
\end{split}
\end{equation*}

\section{Proof of Theorem 3}
\label{app:thm3}
Recall our test is defined as
\begin{equation*}
  \delta_{\mathrm{KL}}(y^{Mn}) = \argmax_{1\le j \le M} \hat{D}_n(Y^{(j)}||\pi).
\end{equation*}

Now we will show the test we proposed is exponentially consistent. By the symmetry of the problem, it is clear that $\mathbb{P}_i\{\delta_{\mathrm{KL}} \ne i \}$ is the same for every $i = 1,\dots,M$, hence
\begin{equation*}
  \max_{i=1,\dots,M} \mathbb{P}_i\{\delta_{\mathrm{KL}} \ne i \} = \mathbb{P}_1\{\delta_{\mathrm{KL}} \ne 1 \}.
\end{equation*}
It now follows
\begin{align*}
    \mathbb{P}_1\{\delta_{\mathrm{KL}} \ne 1 \}& = \mathbb{P}_1\left \{\hat{D}_n(Y^{(1)}||\pi) \le \max _{2\le j\le M}\hat{D}_n(Y^{(j)}||\pi)\right \} \\
      & \le (M-1) \mathbb{P}_1\left \{\hat{D}_n(Y^{(1)}||\pi) \le  \hat{D}_n(Y^{(2)}||\pi)\right \} \\
      & = (M-1) \mathbb{P}_1\left \{\hat{D}_n(Y^{(1)}||\pi) -D(\mu||\pi)+  \hat{D}_n(Y^{(2)}||\pi) \le -D(\mu||\pi)  \right \} \\
      & \le (M-1) \mathbb{P}_1\left \{\Big|\hat{D}_n(Y^{(1)}||\pi) -D(\mu||\pi)\Big|+  \Big|\hat{D}_n(Y^{(2)}||\pi)\Big| \ge D(\mu||\pi)  \right \}\\
      & \le (M-1)\left(\mathbb{P}_1\left \{\Big|\hat{D}_n(Y^{(1)}||\pi) -D(\mu||\pi)\Big| \ge c D(\mu||\pi) \right \} + \mathbb{P}_1\left \{\Big|\hat{D}_n(Y^{(2)}||\pi)\Big|   \ge (1-c)D(\mu||\pi) \right \}\right)
\end{align*}
where $c \in (0,1)$, so that we can optimize over $c$ to get a tighter bound on error exponent.

Now apply the result we proved in Theorem 2. We get
\begin{align}
  \lim_{n\to \infty} -\frac{1}{n} \log\mathbb{P}_1\left \{\Big|\hat{D}_n(Y^{(1)}||\pi) -D(\mu||\pi)\Big| \ge c D(\mu||\pi) \right \} & \le \frac{c^2}{32}\left(\frac{K_1}{K_2}\right)^2 D^2(\mu||\pi)\nn \\
  \lim_{n\to \infty} -\frac{1}{n} \log\mathbb{P}_1\left \{\Big|\hat{D}_n(Y^{(2)}||\pi)\Big|   \ge (1-c)D(\mu||\pi) \right \}& \le \frac{(1-c)^2}{32} D^2(\mu||\pi). \nn
\end{align}
The optimal result is achieved when the two exponents are equal, we get:
\begin{equation*}
  c^*=\frac{K_2}{K_1+K_2},
\end{equation*}
and the error exponent we get is $\frac{1}{32}\left(\frac{K_1}{K_1+K_2}\right)^2 D^2(\mu||\pi) $.
\begin{equation*}
\begin{split}
\alpha(\delta_{\mathrm{KL}},\pi,\mu)\ge \frac{1}{32}\left(\frac{K_1}{K_1+K_2}\right)^2 D^2(\mu||\pi).
\end{split}
\end{equation*}

\section{Proof of Theorem 4}
\label{app:thm4}
We first introduce the McDiarmid's inequality which is useful in bounding the probability of error in our proof.
\begin{lemma}[McDiarmid's Inequality]\label{mcdiarmid}
Let $f:\mathcal{X}^m\rightarrow \mathbb{R}$ be a function such that for all $i\in\{1,\ldots,m\}$, there exist $c_i<\infty$ for which
\begin{equation}\label{bdd}
\underset{X\in\mathcal{X}^m, \tilde{x}\in \mathcal X}{\sup}|f(x_1,\ldots,x_m)-f(x_1,\ldots x_{i-1},\tilde x,x_{i+1},\ldots,x_m)|\leq c_i.
\end{equation}
Then for all probability measure $p$ and every $\epsilon>0$,
\begin{equation*}
\mathbb{P}_X\bigg(f(X)-\mathbb{E}_X(f(X))>\epsilon\bigg)<\exp\left(-\frac{2\epsilon^2}{\sum_{i=1}^mc_i^2}\right),
\end{equation*}
where $X$ denotes $(x_1,\ldots,x_m)$, $\mathbb{E}_X$ denotes the expectation over the $m$ random variables $x_i\thicksim p$, and $\mathbb{P}_X$ denotes the probability over these $m$ variables.
\end{lemma}

In order to analyze the probability of error for the test $\delta_{\mmd}$, without loss of generality, we assume that the first sequence is the anomalous sequence generated by the anomalous distribution $\mu$. Hence,
\begin{flalign*}
\max_{i=1,\dots,M}\mathbb{P}_i\{\delta_{\mmd}\neq i\}&=\mathbb{P}_1(\delta_{\mmd}\neq 1)\nn\\
&=\mathbb{P}_1\bigg(\exists k\neq 1: \mmd_u^2[Y^{(k)},\pi]>\mmd_u^2[Y^{(1)},\pi]\bigg)\nn\\
&\leq \sum_{k=2}^M \mathbb{P}_1\bigg(\mmd_u^2[Y^{(k)},\pi]>\mmd_u^2[Y^{(1)},\pi]\bigg).
\end{flalign*}
For $k=1,\ldots,M$, we have,
\begin{flalign*}
  \mmd_u^2[Y^{(k)},\pi]=\frac{1}{n(n-1)}\sum_{\substack{i,j=1\\i\neq j}}^{n,n}k(y^{(k)}_i,y^{(k)}_j)-\frac{2}{n}\sum_{i=1}^n\mathbb{E}_x[k(y^{(k)}_i,x)]+\mathbb E_{x,x'}[k(x,x')],
\end{flalign*}
where $x$ and $x'$ are i.i.d. generated from $\pi$.
We define function $\Delta_k(Y^{(k)})$,
\[\Delta_k(Y^{(k)})=\mmd_u^2[Y^{(k)},\pi]-\mmd_u^2[Y^{(1)},\pi].\]
It can be shown that,
\[\mE \big\{\mmd_u^2[Y^{(1)},\pi] \big\}=\mmd^2[\mu,\pi],\]
and for $k\neq 1$,
\[\mE \big\{\mmd_u^2[Y^{(k)},\pi]\big\}=0.\]

For $1\leq i\leq n$ and $1\leq k\leq M$, $y^{(k)}_i$ affects $\Delta_k$ through the following terms
\begin{flalign*}
  \frac{1}{n(n-1)}\sum_{\substack{j=1\\j\neq i}}^{n}k(y^{(k)}_i,y^{(k)}_j)-\frac{2}{n}\mathbb{E}_x[k(y^{(k)}_i,x)].
\end{flalign*}
We define $Y^{(k)}_{-i}$ as $Y^{(k)}$ with the i-th component $y^{(k)}_i$ being removed.
Hence, for $1\leq k\leq M$ and $1\leq i\leq n$, we have
\begin{flalign*}
|\Delta_k\big(Y^{(k)}_{-i},y^{(k)}_i\big)-\Delta_k\big(Y^{(k)}_{-i},{y^{(k)}_i}'\big)|\leq \frac{3K}{n}.
\end{flalign*}
i.e., $\Delta_k(Y^{(k)})$ satisfies the bounded difference condition in \eqref{bdd}, with $c_i = \frac{3K}{n}$. Hence, by McDiarmid's inequality,
\begin{flalign*}
  \mathbb{P}_1\bigg(\mmd_u^2[Y^{(k)},\pi]>\mmd_u^2[Y^{(1)},\pi]\bigg)
  &=\mathbb{P}_1\bigg(\Delta_k(Y^{(k)}) > 0\bigg)\\
  &=\mathbb{P}_1\bigg(\Delta_k(Y^{(k)})-\mmd^2[\mu,\pi] > -\mmd^2[\mu,\pi]\bigg)  \\
  &\leq\exp\Big(-\frac{n\mmd^4[\mu,\pi]}{{9K^2}}\Big)
\end{flalign*}
And we prove that
\begin{flalign}
  \alpha(\delta_\mmd,\pi,\mu)\geq \frac{\mmd^4[\mu,\pi]}{{9K^2}}.
\end{flalign}

%
%
%
%



\clearpage

\bibliographystyle{IEEEbib}
\bibliography{SequenceDetection}

\end{document}